\documentclass[conference]{IEEEtran}
\IEEEoverridecommandlockouts
\ifCLASSINFOpdf
 \usepackage[pdftex]{graphicx}
  \DeclareGraphicsExtensions{.pdf,.jpeg,.png}
\else
\fi
\usepackage{amsthm,amsfonts,amssymb}
\usepackage{amsmath}
\usepackage{bm}

\usepackage[font=footnotesize]{subfig}
\newcommand{\Toro}[2]{\Phi_{\pmb{#1}}(#2)}

\newcommand{\CurvaT}[1]{\pmb{s}_{T_{#1}}}

\newcommand{\R}{\mathbb{R}}

\newcommand{\Vo}{\mathcal{V}}
\newcommand{\Sc}{\mathcal{S}\mathcal{C}}
\newcommand{\PP}{\mathcal{P}}
\newcommand{\cc}{\pmb{c}}
\newcommand{\bb}{\pmb{b}}

\newcommand{\uu}{\pmb{u}}

\newtheorem{teo}{Theorem}

\newtheorem{ex}{Example}
\newtheorem{prop}{Proposition}

\begin{document}
%
\title{Curves on torus layers and coding for continuous alphabet sources}

\author{\IEEEauthorblockN{Antonio Campello\IEEEauthorrefmark{1}}
\IEEEauthorblockA{Institute of Mathematics, Statistics\\and Computer Science\\
University of Campinas, S\~ao Paulo \\
13083-859, Brazil\\
Email:  campello@ime.unicamp.br}
\and
\IEEEauthorblockN{Cristiano Torezzan\IEEEauthorrefmark{2}}
\IEEEauthorblockA{School of Applied Sciences\\
University of Campinas, S\~ao Paulo\\
13484-350, Brazil \\
Email: cristiano.torezzan@fca.unicamp.br}
\and
\IEEEauthorblockN{Sueli I. R. Costa\IEEEauthorrefmark{3}\thanks{Work partially supported by
FAPESP\IEEEauthorrefmark{1} under grant 2009/18337-6, CNPq\IEEEauthorrefmark{3}
309561/2009-4 and FAPESP 2007/56052-8}}
\IEEEauthorblockA{Institute of Mathematics, Statistics\\and Computer Science\\
University of Campinas, S\~ao Paulo \\
13083-859, Brazil\\
Email:  sueli@ime.unicamp.br}}


%


\maketitle

\begin{abstract}
In this paper we consider the problem of transmitting a continuous alphabet 
discrete-time source over an AWGN channel. The design of good curves for this purpose relies on geometrical properties of spherical codes and projections of $N$-dimensional lattices. We propose a constructive scheme based on a set of  curves on the surface of a $2N$-dimensional sphere and present comparisons with some previous works. 
\end{abstract}
\section{Introduction}
The problem of designing good codes for continuous alphabet sources to be transmitted over a channel with power constraint can be viewed as the one of constructing curves in the Euclidean space of maximal length and such that its folds (or laps) are a good distance apart. When the channel noise is bellow a certain threshold, a bigger length essentially means a higher resolution when estimating the sent value. On the other hand, if the folds of the curve come too close, this threshold will be small and the mean squared error (mse) will be dominated by larger errors as consequence of decoding to the wrong fold. Explicit constructions of curves for analog source-channel coding were presented, for example, in \cite{Sueli} and \cite{polynomial}. 

In this work, we will consider spherical curves in the $2N$-dimensional Euclidean space. We develop an extension of the construction presented in \cite{Sueli} to a set of curves on layers of flat tori. Our approach explores geometrical properties of spherical codes and projections of $N$-dimensional lattices. In the scheme presented here, homogeneity and low decoding complexity properties were preserved whereas the total length can be meaningfully increased.

This paper is organized as follows. In Section II, we introduce some mathematical background used in our approach. In Section III we state the problem, while in Section IV we present a scheme to design piecewise homogeneous curves on the Euclidean sphere and describe the encoding/decoding process. In Section V we derive a scaled version of the $\textit{Lifting Construction}$ \cite{FatStrut} suitable to our problem and present some examples and length comparisons with some previous constructions.

\section{Background}
\subsection{Flat Tori}
The unit sphere $S^{2N-1} \subset \R^{2N} $ can be foliated with flat tori \cite{BergerGostiaux, Cristiano} as follows. For each unit vector $\cc = (c_{1},c_{2},..,c_{N}) \in \R^{N}, c_i > 0$, and $\pmb{u}=(u_1,u_2,\ldots,u_N) \in \R^N$, let $\Phi_{\cc}:\R^N \rightarrow \R^{2N}$ be defined by
\begin{equation}
\small{
\Phi_{\cc} (\pmb u)=\left(c_{1}\left(\cos \frac{u_{1}%
}{c_{1}},\sin \frac{u_{1}}{c_{1}}\right),\dots,c_{N}\left(\cos \frac{u_{N}}{c_{N}},\sin \frac{u_{N}}{c_{N}}\right)\right).
}
\label{eq:Toro}
\end{equation}
This periodic function $\Phi _{\cc}$ is a local isometry on its image, the torus $ T_{\cc}$, a flat $N$-dimensional surface contained in the unit sphere $S^{2N-1}$. $ T_{\cc} = \Phi_{\cc}(R^N)$ is also the image of the hyperbox:
\begin{equation}
\label{para}
\PP_{\cc} := \{\uu \in \R^{N}; 0 \leq u_{i} < 2 \pi c_{i}\}, \ \ 1\leq i\leq N.
\end{equation}

Note also that each vector of $S^{2N-1}$ belongs to one, and only one, of these flat tori if we consider also the degenerated cases where some $\cc_i$ may vanish. Thus we say that the family of flat tori $T_{\cc}$ and their degenerations, with $\cc = (c_{1},c_{2},..,c_{N})$, $ \left\|  \cc  \right\|  =1$, $c_{i}  \geq 0$, defined above is a foliation on the unit sphere of $S^{2N-1}\subset \R^{2N}.$

It can be shown (Proposition 1 in \cite{Cristiano}) that the minimum distance between two points, one in each flat torus $T_{\bb}$ and $T_{\cc}$, is 
\begin{equation}
\label{eq:distdoistoros}
d(T_{\cc},T_{\bb})= \left\|  \cc- \bb \right\| = \left( \sum_{i=1}^N (c_i - b_i)^2\right)^{1/2}.
\end{equation}

The distance between two points on the same torus $T_{\pmb{c}}$ given by
\begin{equation*}
d(\Phi_{\pmb{c}}(\pmb{u}),\Phi_{\pmb{c}}(\pmb{v}))=2\sqrt{\sum c_{i}^{2}\sin^{2}(\frac{u_{i}-v_{i}%
}{2c_{i}})}\label{SameTorus}%
\end{equation*}
is bounded in terms of $\|\pmb{u}-\pmb{v}\|$ by the following proposition.
\begin{prop}\cite{Cristiano}
Let $\pmb{c=}(c_{1},c_{2},..,c_{N})$, $ \left\|  \cc  \right\|  =1$,
and let  $\pmb{u}, \pmb{v} \in \PP_{\cc}$. Let $\Delta =  \left\|\pmb{u}-\pmb{v}\right\|$ and $\delta= \left\|\Phi_{\pmb{c}}(\pmb{u})- \Phi_{\pmb{c}}(\pmb{v}) \right\|$.
Then
\begin{equation}
\displaystyle 
\frac{2}{\pi}\Delta\leq\dfrac{\sin\frac{\Delta}{2c_{\xi}}}{\frac{\Delta}{2c_{\xi}%
}}\Delta\leq\delta\leq\dfrac{\sin\frac{\Delta}{2}}{\frac{\Delta}{2}}\Delta
\leq\Delta
\label{eq:FamousRelation}
\end{equation}
where $\displaystyle c_{\xi} = \min c_i$.
\\
\end{prop}

\subsection{Curves}

As pointed out in \cite{Sueli}, some important properties of a curve from a communication point of view are its \textit{stretch} and \textit{small-ball radius}. Given a curve $\pmb{s}: [a,b] \mapsto \mathbb{R}^N$, the Voronoi region $V(x)$ of a point $\pmb{s}(x)$ is the set of all points in $\mathbb{R}^N$ which are closer to $\pmb{s}(x)$ than to any other point of the curve. If $H(x)$ denotes the hyperplane orthogonal to the curve at $\pmb{s}(x)$, then the maximal \textit{small-ball radius} of $\pmb{s}$ is the largest $r > 0$ such that $B_r(\pmb{s}(x)) \cap H(x) \subset V(x)$ for all $x \in [a,b]$, where $B_r(\pmb{s}(x))$ is the Euclidean ball of radius $r$ centered at $\pmb{s}(x)$. Intuitively this means that the ``tube'' of radius $r$ placed along the curve does not intersect itself. The \textit{stretch} $\mathcal{S}(x)$ is the function $\left\| \dot{\pmb{s}}(x) \right\|$ where $\dot{\pmb{s}}(x)$ is the derivative of $\pmb{s}(x)$. If $\left\|\dot{\pmb{s}}\right\|$ does not depend on $x$ (as it is the case of the curves considered here) we will refer to the stretch as simply $\mathcal{S}$. The length of a curve is given by $\int_a^b \mathcal{S}(x) dx$. In this paper we will be interested in curves with large length and small-ball radius.

\section{Problem statement}

The underlying communication system we consider here is illustrated in Figure \ref{fig:com}. Given an input real value $x$, within the unit interval $[0,1]$, the encoder maps $x$ into a point $\pmb{s}(x)$ of a curve $\pmb{s}$ on $S^{2N-1}$, which will be sent over an AWGN channel. By properly scaling the curve we guarantee that the transmitted energy is $\alpha^2$. The decoder will then compute an estimate for the sent value while trying to minimize the mean squared error (mse) $E[(X-\hat{X})^2]$ of the process. 

For the torus layers scheme, the unit interval will be partitioned into $M$ intervals of different length, and each of them mapped into a curve on one of the layers. It is worth noticing that, for the special case $M=1$, if we choose the torus associated to the vector $c = \frac{1}{\sqrt{N}} (1,\ldots,1)$ to encode the information, then the scheme proposed here is exactly the one analysed in \cite{Sueli}. However, we are interested in the analysis for $M > 1$, in which case the curves to be presented outperform the ones presented in \cite{Sueli} (see also \cite{FatStrut}) in terms of the tradeoff between length and small-ball radius. \\

\begin{figure}[!htb]
\includegraphics[scale=0.65]{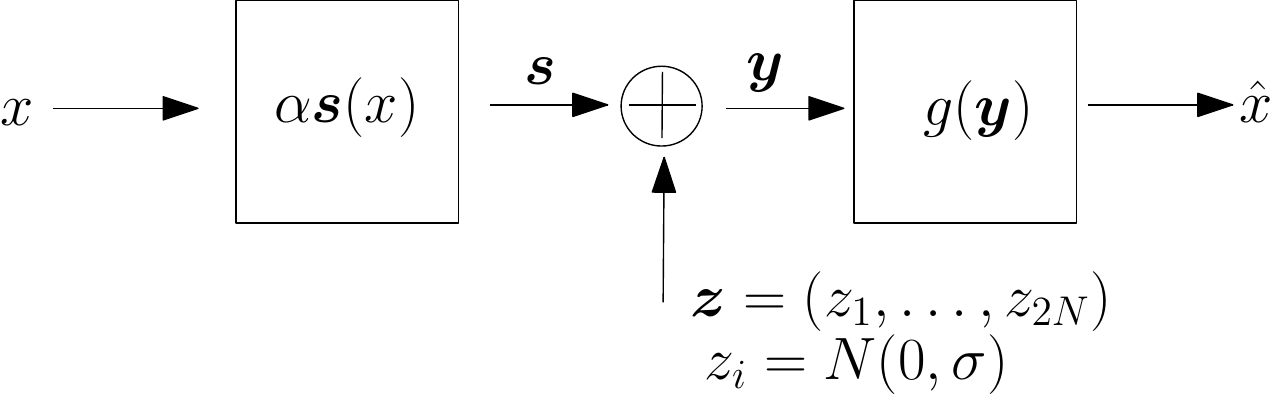}
\caption{Communication system}	
\label{fig:com}
\end{figure}

The design of those curves in the next section is essentially divided in two parts. First, we choose a collection of tori on the surface of the Euclidean sphere $S^{2N-1}$ at least $\delta$ apart. The approach for doing this is via discrete spherical codes in $\mathbb{R}^N$. Second, we show a systematic way of constructing curves on each layer, via projection-lattices in $ \mathbb{R}^{N-1}$. Finally, we give a description of the whole signal locus and summarize the encoding process.

\section{Our approach}

\subsection{Torus layers}
\label{sec:torusLayers}
Given a fixed small ball radius $\delta$, the first step of our approach is to define a collection $T=\left\{T_1, T_2, \cdots, T_M \right\}$ of flat tori on $S^{2N-1}$ such that the minimum distance \eqref{eq:distdoistoros} between any two of them is greater than $2\delta$. This step is equivalent to design a $N$-dimensional spherical code $\Sc$ with minimum distance $2\delta$ and consider just the points with non-negative coordinates.

We denote this sub-code by $$\Sc_+ = \left\{c  \in \Sc, c_i \geq 0, \ \ 1 \leq i \leq N  \right\}.$$
Each point $\cc \in \Sc_+$ defines a hyperbox $\PP_{\cc}$ (\ref{para}) and hence a flat torus $T_{\cc}$ in the unit sphere $S^{2N-1}$. 

There are several ways of constructing spherical codes that can be employed here, e.g. \cite{eric,Hamkins1} or even on layers of flat tori as introduced in \cite{Cristiano}.

\subsection{Curves on each torus}
Let $T_{\pmb{c}}$ be a torus represented by the vector $\pmb{c} \in \Sc_+$ as defined in the previous section. On the surface of $T_{\pmb{c}}$ we will consider curves of the form: 
\begin{equation}
\CurvaT{c}(x) = \Toro{c}{2\pi \hat{\pmb{u}} x},
\label{def:Curva}
\end{equation}
where $C = \mbox{diag}(c_1,\ldots,c_L)$, $\hat{\pmb{u}} = \pmb{u}C = (c_1 u_1, \ldots, c_N u_N)$, $\Phi_{\pmb{c}}$ is given by \eqref{eq:Toro} and $x \in [0,1]$.

Provided that $\pmb{u} \in \mathbb{Z}^N, \gcd(u_i) = 1$, those curves are closed (a type $(u_1,\ldots,u_N)$-knot), and due to periodicity and local isometry properties of $\Phi_{\pmb{c}}$ their lengths are $2 \pi \left\| \hat{\pmb{u}} \right\|$. They are also the image through $\Phi_{\pmb{c}}$ of the intersection between the set of lines $W = \left\{ \hat{\pmb{u}}x + \hat{\pmb{n}} : \hat{\pmb{n}} = \pmb{n} C, \pmb{n} \in \mathbb{Z}^{L} \right\}$ and the box $\PP_{\cc}$. For \linebreak $c = \frac{1}{\sqrt{N}}(1,\ldots,1)$, these are exactly the curves analysed in \cite{Sueli}.

Let $r_{\pmb{c}}(\pmb{u})$ be the minimum distance between two different lines in $W$, then we have
\begin{equation}
\begin{split}
r_{\pmb{c}}(\pmb{u}) &:= \min_{\hat{\pmb{n}} \neq k \pmb{\hat{u}}, k \in \mathbb{Z}} \min_{\hat{x}, x} \left\| \hat{\pmb{u}}x - (\hat{\pmb{u}}\hat{x}+\hat{\pmb{n}}) \right\| \\
& = \min_{\hat{\pmb{n}} \neq k \pmb{\hat{u}}, k \in \mathbb{Z}} \min_{x} \left\| \hat{\pmb{u}}x - \hat{\pmb{n}} \right\| \\
& = \min_{\hat{\pmb{n}} \neq k \pmb{\hat{u}}, k \in \mathbb{Z}} \left\| P_{\hat{\pmb{u}}^{\perp}} (\hat{\pmb{n}}) \right\| \\
& = \min_{\hat{\pmb{n}} \notin \hat{\pmb{u}}^{\perp}} \left\| P_{\hat{\pmb{u}}^{\perp}} (\hat{\pmb{n}}) \right\|,
\end{split}
\end{equation}
where $P_{\hat{\pmb{u}}^{\perp}} (\hat{\pmb{n}})$ denotes the orthogonal projection of $\hat{\pmb{n}}$ onto the hyperplane $\hat{\pmb{u}}^{\perp}$ which is given by the standard projection formula

\begin{equation}
P_{\hat{\pmb{u}}^{\perp}} (\hat{\pmb{n}}) = \hat{\pmb{n}} \left(I_N - \frac{\pmb{\hat{u}}^t \pmb{\hat{u}}}{\pmb{\hat{u}} \pmb{\hat{u}}^t}\right).
\label{eq:proj}
\end{equation}

Let  $\Lambda_{\pmb{c}} = c_1 \mathbb{Z} \oplus \ldots \oplus c_N \mathbb{Z}$ be the rectangular lattice generated by matrix $C$, then  $ r_{\pmb{c}}(\pmb{u})$ is the length of shortest non-zero vector of the projection\footnote{In general, the projection of a lattice $\Lambda$ onto a subspace $H$ is \textit{not} a lattice unless certain special conditions are met, e.g., when $H^\perp$ is spanned by primitive vectors of $\Lambda$ \cite{Perfect}. This will be always the case in this paper, since $H = \hat{\pmb{u}}^\perp$ for a primitive vector $\hat{\pmb{u}}$.} of $\Lambda_{\pmb{c}}$ onto $\hat{\pmb{u}}^{\perp}$. Due to Equation \eqref{eq:FamousRelation}, the small-ball radius $\delta_{\pmb{u},\pmb{c}}$ of $\pmb{s}_{T_{\pmb{c}}}$ can be bounded in terms of $r_{\pmb{c}} (\pmb{u})$ as follows:

\begin{equation}
2 c_{\xi} \sin \left({\frac{\pi  r_{\pmb{c}}(\pmb{u})}{2 c_{\xi}}}\right) \leq \delta_{\pmb{u},\pmb{c}} \leq 2 \sin \left({\frac{\pi r_{\pmb{c}}(\pmb{u})}{2}}\right),
\label{eq:distortionCurve}
\end{equation}
where $c_{\xi} = \displaystyle \min_{i}{c_i}$ and $c_i > 0$. Thus, for small values of $\delta_{\pmb{u}, \pmb{c}} $, we have $\delta_{\pmb{u},\pmb{c}} \approx  \pi r_{\pmb{c}}(\pmb{u})$. Our goal is to choose $\pmb{u}$ in order to maximize $r_{\pmb{c}}(\pmb{u})$. In addition, we also want to reach a contrary objective, which is the one of maximizing the arc length \linebreak $l_{\pmb{u},\pmb{c}} = 2 \pi \left\|\hat{\pmb{u}} \right\|$ of $\pmb{s}_{T_{\pmb{c}}}$. 

It is possible to show (Proposition 1.2.9. in \cite{Perfect}) that the density of the lattice $P_{\hat{\bm{u}}^{\perp}} (\Lambda_{\bm{c}})$, the projection of $\Lambda_{\bm{c}}$ onto $\hat{\bm{u}}^{\perp}$, is given by:

\begin{equation}
\Delta(P_{\hat{\bm{u}}^{\perp}} (\Lambda_{\bm{c}})) = \frac{r_{\pmb{c}}(\pmb{u})^{N-1} \left\|\hat{\pmb{u}} \right\|}{2^{N-1} \prod_{i=1}^N c_i} \leq \frac{\Delta_{N-1}}{\Vo_{N-1}}
\label{eq:BoundDensidade}
\end{equation}
where $\Delta_{N-1}$ is the density of the best lattice in dimension $(N-1)$ and $\Vo_{N-1}$ is the volume of the $(N-1)$-dimensional unit sphere. For the case when all entries of $\bm{c}$ are equal, $\Lambda_{\bm{c}}$ is equivalent to $\mathbb{Z}^N$ and it was shown in \cite{FatStrut} that we can make the above bound as tight as we want. We will show in Section V that this is also true for an arbitrary $\bm{c}$ i.e., that projections of the rectangular lattice $\Lambda_{\pmb{c}}$ can also yield dense lattice packings and therefore we can construct curves on the flat torus with the parameters arbitrary close to this bound.

\begin{ex} Let $N = 2$. Consider the local isometry
\begin{equation}
\Phi_{\cc} (\pmb u)=\left(c_{1} \cos \frac{u_{1}
}{c_{1}}, c_1 \sin \frac{u_{1}}{c_{1}} ,c_{2}\cos \frac{u_{2}}{c_{2}},c_2 \sin \frac{u_{2}}{c_{2}}\right)
\end{equation}
on the flat torus $T_{\pmb{c}}$ and the line segment given by \linebreak $\pmb{v}(x) = x \pmb{v}$, $= x (2 \pi u_1 c_1, 2\pi u_2 c_2),u_1, u_2 \in \mathbb{Z}$ and $0 \leq x \leq 1$. The curve $\CurvaT{c}(x)$ will be the composition $\Phi(\pmb{v}(x))$ and we have
\begin{equation} 
r_{\pmb{c}}(\pmb{v}) \left\| \pmb{v} \right\| = 2\pi c_1 c_2 \Rightarrow  r_{\pmb{c}}(\pmb{v}) = \frac{c_1 c_2}{\sqrt{ w^2 c_1^2+ (w+1)^2 c_2^2}}
\end{equation}

This curve in $\mathbb{R}^4$ will turn around $u_1$-times the circle obtained by its projection on the first two coordinates, whereas turning around $u_2$-times the circle of radius $c_2$ given by its last two coordinates (a type $(u_1,u_2)$ knot on the flat torus $T_{\pmb{c}}$). For this case, we can calculate the exact small-ball radius $\delta$ as $d(\Phi( \alpha \pmb{v}^{\perp} ), \Phi(0))$ where $\pmb{v}^{\perp} = (-2\pi c_2 u_2, 2 \pi c_1 u_1)$ and $\alpha = \pi c_1 c_2 / \left\| \pmb{v} \right\|^2$. In Figure \ref{fig:encodingprocess} it is illustrated the curve $(u_1,u_2)=(4,5)$, with $c_1 > c_2$.
\end{ex}

\subsection{Encoding}
Let $T = \left\{ T_1, \ldots, T_M \right\}$ be a collection of $M$ tori as designed in Section \ref{sec:torusLayers}. For each one of these tori, let $\CurvaT{k}(x) = \Phi(2\pi \hat{\pmb{u}}_kx), (k=1,2,\ldots,M)$ be the curve on $T_k$, determined by the vector $\hat{\pmb{u}}_k$ \eqref{def:Curva} and  consider $L = \sum_{j=1}^M l_j$, where $l_k$ is the length of $\CurvaT{k}$.

Now split the unit interval $[0,1]$ into $M$ pieces according to the length of each curve:
\begin{equation*}
\left[0,1\right) = I_1 \cup I_2, \ldots \cup I_M\mbox{, where}
\end{equation*}
\begin{equation*}
I_k = \left[\frac{\sum_{j=1}^{k-1} l_j}{L},\frac{\sum_{j=1}^{k} l_j}{L} \right), \mbox{ for } k = 1, \ldots, M.
\end{equation*}
and consider the bijective mapping
$$\begin{array}{cc}
f_k: I_k  \to [0,1) \\
f_k(x) = \displaystyle\frac{x - \sum_{j=1}^{k-1} l_j/L}{l_k/L}.
\end{array}
$$
Then the full encoding map $\pmb{s}$ can be defined by
\begin{equation}
\pmb{s}(x) := \CurvaT{k}(f_k(x)), \mbox{ if } x \in I_k.
\label{eq:encoding}
\end{equation}

The stretch of $\pmb{s}$ will be constant and equal its total length $L$ and the small-ball radius of $\pmb{s}$ is the minimum small-ball radius $\delta$ of the curves $\CurvaT{k}$, provided that the distance between any pair of torus in $T$ is at least $2\delta$.


To encode a value $x$ within $[0,1]$ we apply the map \eqref{eq:encoding}. The signal locus will be a set of $M$ closed curves, each one lying on a torus layer $T_k$ and defined by a vector $\pmb{u}_k$. This whole process is illustrated in Figure \ref{fig:encodingprocess}.

\begin{figure}[!htb]
\centering
\includegraphics[scale=0.7]{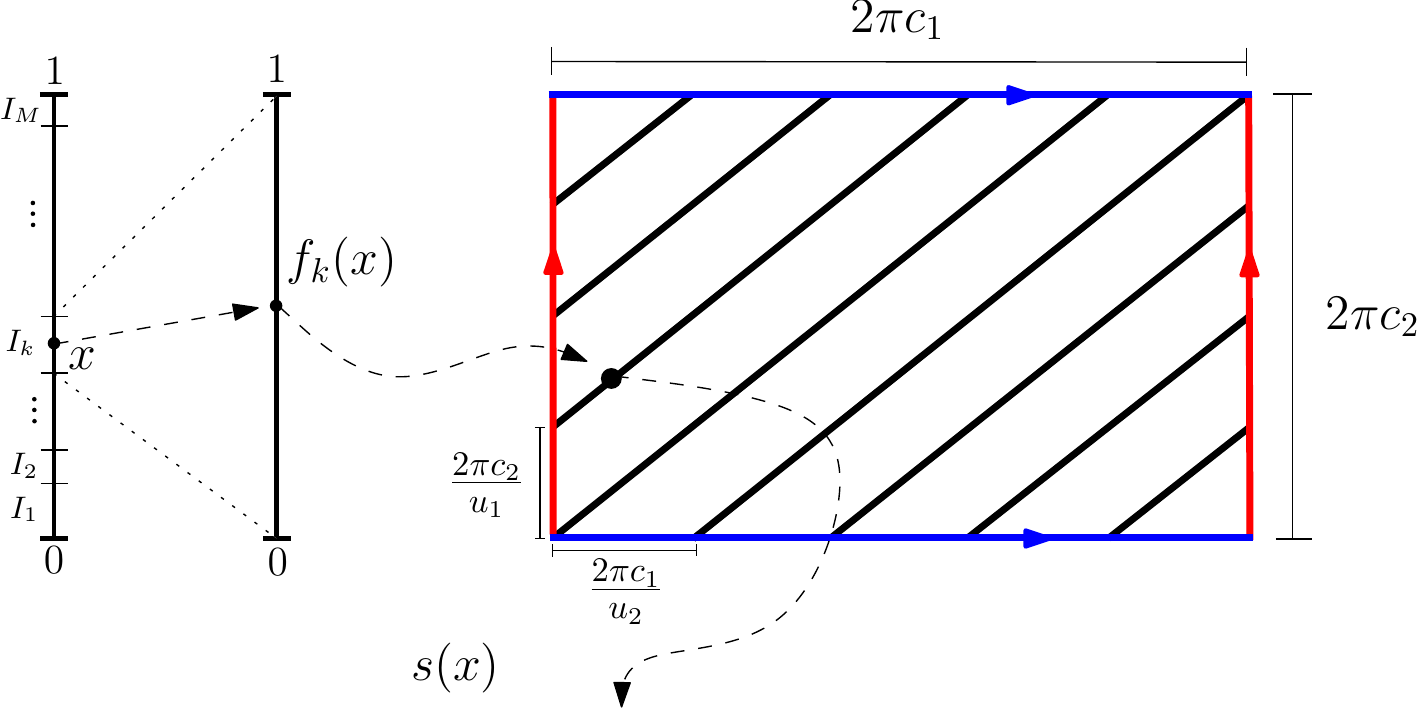}
\includegraphics[scale=0.18]{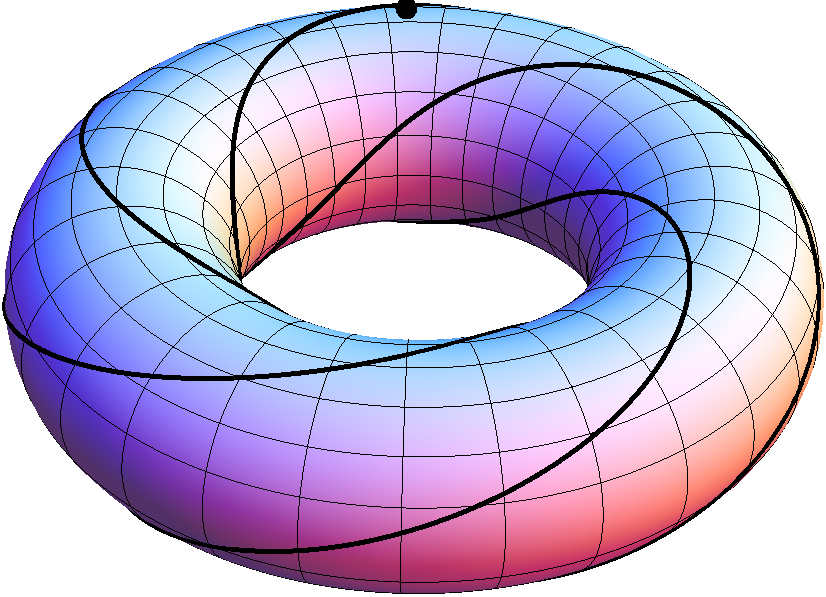}
\caption{Encoding Process}
\label{fig:encodingprocess}
\end{figure}

If the source is uniformly distributed over $[0,1]$, the encoding process presented above is a proper one, since all subintervals will be equally stretched. For other applications, however, it could be worth considering another partition.
\subsection{Decoding}
Given a received vector $\pmb{y} \in \mathbb{R}^{2N}$, the maximum likelihood decoding is finding $\hat{x}$ such that:

$$\hat{x} = \mbox{arg}\min_{x \in [0,1]} \left\|\pmb{y} - \pmb{s}(x)\right\|.$$

Since exactly solving this problem is computationally expensive we focus on a suboptimal decoder.

For $ 0 \neq \gamma_i = \sqrt{y_{2i-1}^2+y_{2i}^2} $, we can write
{\small
\begin{eqnarray}
\pmb{y} & = &  \left( \gamma_1 \left( \frac{y_1}{\gamma_1}, \frac{y_2}{\gamma_1} \right), \hdots , \gamma_N \left( \frac{y_{2N-1}}{\gamma_N}, \frac{y_{2N}}{\gamma_N} \right)
 \right) \nonumber \\
\pmb{y} & = &  \left( \gamma_1\left(\cos{\frac{\theta_{1}}{\gamma_1}}, \sin{\frac{\theta_{1}}{\gamma_1}} \right), \hdots ,
\gamma_N \left(\cos{\frac{\theta_{N}}{\gamma_N}},\sin{ \frac{\theta_{N}}{\gamma_N}} \right) \right), \nonumber 
\end{eqnarray}
}
where,
 \begin{eqnarray}
  \theta_{i} & = & \arccos{ \left( \dfrac{y_{2i-1}}{\gamma_i} \right) } \gamma_i,  \ \ 1 \leq i \leq N \nonumber.
 \end{eqnarray}

The process of finding the closest layer involves a $N$-dimensional spherical decoding of $\pmb{\gamma}=(\gamma_1,\gamma_2, \cdots, \gamma_N)$, which has complexity $\linebreak \mathcal{O}(MN)$.

Let $\pmb{c}_i = (c_{i1}, c_{i2}, \cdots, c_{iN})$ be the closest point in $\in \Sc_+$ to $\pmb{\gamma}$ and
$$\bar{\pmb{y}}_i = \left( c_{i1}\left(\cos{\frac{\theta_{1}}{\gamma_1}}, \sin{\frac{\theta_{1}}{\gamma_1}} \right), \hdots ,
c_{iN} \left(\cos{\frac{\theta_{N}}{\gamma_N}},\sin{ \frac{\theta_{N}}{\gamma_N}} \right) \right)$$ be the projection of $\pmb{y}$ in the torus $T_{c_i}$, i.e., 
$$
||\pmb{y}-\bar{\pmb{y}}_i|| \leq || \pmb{y} - \pmb{v} || \, ,  \forall \, \pmb{v} \in T_{c_i}.
$$
From now on, we proceed the process by using a slight modification of the torus decoding algorithm \cite{Sueli} applied to the $N$-dimensional hyperbox $\mathcal{P}_{\cc_i}$. The complexity of this algorithm is given by $\mathcal{O}(N \left\| \pmb{u}_i \right\|_1)$, where $\pmb{u}_i$ is the vector that determines the curve $\CurvaT{i}$. Hence, if $M = \mathcal{O}( \max_i{\left\| \pmb{u}_i \right\|_1)}$, the overall complexity of the process described in this section will be $\mathcal{O}(N \max_i{\left\| \pmb{u}_i \right\|_1)}$, the same as for the torus decoding.

\section{A Scaled Lifting Construction}
\subsection{The construction}
The \textit{Lifting Construction} was proposed in \cite{FatStrut} as a solution to the problem of finding dense lattices which are equivalent to orthogonal projections of $\mathbb{Z}^N$ along integer vectors (``fat-strut'' problem). In this section we adapt that strategy in order to construct projections of the lattice $\Lambda_{\pmb{c}}$ which approximate any $(N-1)$ dimensional lattice (hence the densest one) with the objective of finding curves in $\mathbb{R}^N$ approaching the bound \eqref{eq:BoundDensidade}. We adopt here the lattice terms as in \cite{SloaneLivro}. For our purposes, the proximity measure for lattices will be the distance between their Gram matrices, as in \cite{FatStrut}. This notion measures how close a lattice is to another one up to congruence transformations (rotations or reflections).

The dual of a lattice $\Lambda \in \mathbb{R}^N$ is the set $\Lambda^* = \{ \pmb{x} \in \mbox{span}(\Lambda): \langle \pmb{x}, \pmb{y} \rangle \in \mathbb{Z} \,\, \forall \pmb{y} \in \Lambda\}$ where $\mbox{span}(\Lambda)$ is the subspace spanned by a basis of $\Lambda$. Now let $\Lambda_{\pmb{c}} = c_1 \mathbb{Z} \oplus \ldots \oplus c_N \mathbb{Z}$ be the rectangular lattice generated by the diagonal matrix $C$ (and with Gram matrix $CC^t = C^2$). By scaling $\Lambda_{\pmb{c}}$, we can assume that $c_1 = 1$. With this condition, if $P_{\pmb{u}^{\perp}}(\Lambda_{\pmb{c}})^*$ denotes the dual of the projection of $\Lambda_{\pmb{c}}$ onto a vector $\hat{\pmb{u}}=(1, u_2 c_2 \ldots, u_N c_N)$ ($u_i \in \mathbb{Z}$), then a generator matrix for $P_{\hat{\pmb{u}}^{\perp}}(\Lambda_{\pmb{c}})^*$ is given by

\begin{equation}
M = \begin{pmatrix}-u_2 & 1/c_2 & 0 & \ldots & 0 \\
-u_3 & 0 & 1/c_3 & \ldots & 0 \\
\vdots & \vdots & \vdots & \ddots & \vdots \\
-u_n & 0 & 0 & \ldots & 1/c_n.
\end{pmatrix}
\label{matrizBoa},
\end{equation}
what can be derived as a consequence of Prop. 1.3.4 \cite{Perfect}. In what follows we derive a general construction of projections such that $P_{\pmb{u}^{\perp}}$ is arbitrarily close to a target lattice $\Lambda \in \mathbb{R}^{N-1}$.

\begin{teo} Let $\Lambda_{\pmb{c}} = \mathbb{Z} \oplus c_2 \mathbb{Z} \ldots \oplus c_N \mathbb{Z}$, $c_i \in \mathbb{R}$. Let $\Lambda$ be a target lattice in $\mathbb{R}^{N-1}$ and consider a lower triangular generator matrix  $L^* = (l_{ij}^*)$ for $\Lambda^*$. If $\Lambda_w^*$, $w \in \mathbb{N}$ is the sequence of lattices generated by the matrices

\begin{equation}
L_w^* := \left(\begin{array}{ccccc}
\lfloor{w l_{11}^*}\rfloor & \frac{1}{c_2} & \ldots & \ldots & 0 \\
\lfloor{w l_{21}^*} \rfloor & \frac{\lfloor{w l_{22}^* c_2}\rfloor }{c_2} & \ldots & \ldots & 0 \\
\vdots & \vdots & \ddots & \ldots & 0 \\
\lfloor{w l_{n1}^*} \rfloor & \frac{\lfloor{w l_{n2}^* c_2} \rfloor }{c_2} & \ldots & \frac{\lfloor{c_{n-1} w l_{nn}^*} \rfloor}{c_{n-1}} & \frac{1}{c_n}
\end{array}\right),
\label{matrizCons}
\end{equation}
then:
\begin{enumerate}
\item[(i)] $L_w^* = P_{\hat{\pmb{u}}^{\perp}}(\Lambda)^*$ for some $\hat{\pmb{u}} \in \mathbb{R}^N$ and
\item[(ii)] $(1/w^2) L_w^* L_w^{*t} \to L^* L^{*t}$ as $w \to \infty$.
\end{enumerate}
In other words, for large $w$, $\Lambda$ will be approximated (in the sense of the Gram matrices) by projections of $\Lambda_{\pmb{c}}$.
\end{teo}
\begin{proof} Through applying elementary (integer) operations on $L_w^*$ we can put it on form \eqref{matrizBoa} for some integers $u_2,\ldots,u_n$ depending on $w$, hence $L_w^*$ is a generator matrix for $P_{\hat{\pmb{u}}^{\perp}}(\Lambda)^*$, proving the first statement.
For the second statement, we clearly have $(1/w) L_w^* \to \left[ L^* \,\,\,\,\, \pmb{0} \right]$ as $w \to \infty$, where $\pmb{0}$ is the $(n-1) \times n$ all-zero column vector. Therefore, $(1/w^2) L_w^* L_w^{*t} \to \left[ L^* \,\,\,\,\, \pmb{0} \right]\left[ L^* \,\,\,\,\, \pmb{0} \right]^t = L^* L^{*t}$.
\end{proof}

\begin{ex} Consider the hexagonal lattice \cite{SloaneLivro}, which is the best packing in two dimensions and is equivalent to its dual. One of its generator matrix is

$$L^* = \left(\begin{array}{cc} 1 & 0 \\ \frac{1}{2} & \frac{\sqrt{3}}{2} \end{array} \right).$$

We apply the construction above and reduce, through elementary operations, the matrix $L_w^*$ in order to put it on form \eqref{matrizBoa}. After re-scaling the rectangular lattice $\Lambda_{\pmb{c}}$, we find the sequence of vectors

$$\hat{\pmb{u}}_w = (c_1, -2wc_2, (2w \lfloor w \sqrt{3} c_2/c_1 \rfloor - w) c_3).$$ 
The projections of $c_1 \mathbb{Z} \oplus c_2 \mathbb{Z} \oplus c_3 \mathbb{Z}$ onto $\hat{\pmb{u}}_w^\perp$ will be, up to equivalence, arbitrarily close to $A_2$ when $w \to \infty$.
\end{ex}

\subsection{Comparisons: Curves in $\mathbb{R}^6$}
\label{exemplo}
Here we compare our approach of construct curves on torus layers with the curves constructed in \cite{FatStrut}  and \cite{Sueli} in terms of length for given small-ball radii. Given $\delta > 0$, we first consider a set of flat tori associated to a spherical code in $R^3$, with minimum distance greater than $2 \delta$, as described in Section \ref{sec:torusLayers}. Through the first inequality of \eqref{eq:distortionCurve}, for each torus, we can find $r_{\pmb{c}}$ in order to guarantee that the curves on the flat tori will have small-ball radius at least $\delta$ (this is also done in the case of the curves on the torus $c = 1/\sqrt{L}(1,1,1)$). We then look for the larger element of the sequence of vectors that produces a projection with minimum distance at least $r_{\pmb{c}}$. 

\begin{figure}[htb!]
\centering
\includegraphics[scale=0.55]{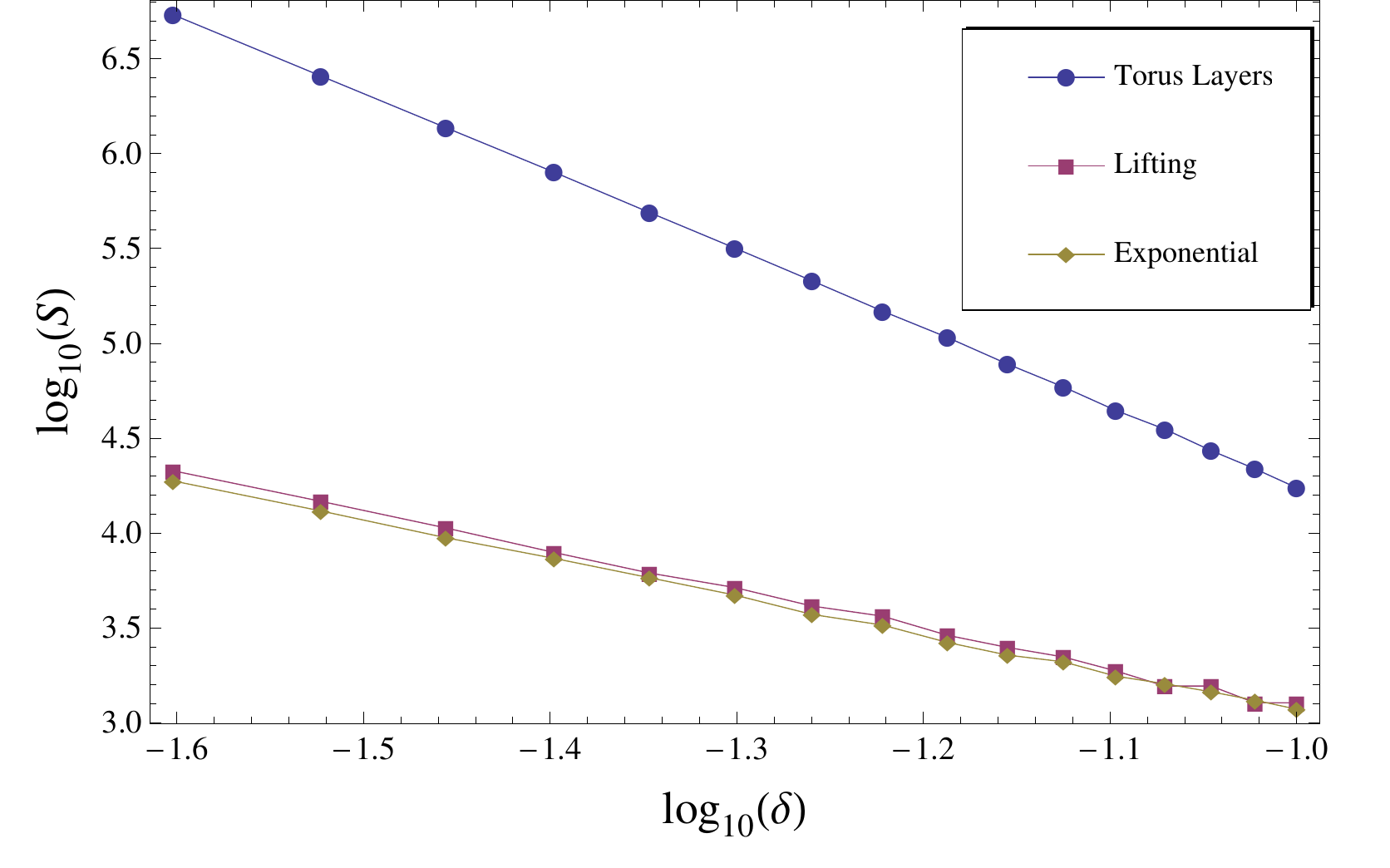}
\caption{Comparison between diferent approaches in terms of small-ball radius $\delta$ and arc-length $\mathcal{S}$. }
\label{grafico}
\end{figure}

In Figure \ref{grafico}, the first curve from the bottom to the top represents the exponential sequence \cite{Sueli}. The second one is obtained by directly applying the Lifting Construction \cite{FatStrut} to the hexagonal lattice and the last one displays the total length associated to our scheme.

\section{Conclusion}
The problem of transmitting a continuous alphabet source over an AWGN channel was considered through an approach based on curves designed in layers of flat tori on the surface of a $(2N)$-dimensional Euclidean sphere. This approach explores connections with constructions of spherical codes and is related to the problem of finding dense projections of the lattice $c_1 \mathbb{Z} \oplus \ldots \oplus c_N \mathbb{Z}$.

This work is a generalization of both the scheme proposed in \cite{Sueli} and the Lifting Construction in \cite{FatStrut}. As a consequence, our scheme compares favorably to previous works in terms of the tradeoff between total length and small-ball radius, which is a proper figure of merit for this communication system.

In spite of the improvements in terms of length versus small-ball radius, the constructiveness, homogeneity and overall complexity of the decoding algorithm are features preserved from \cite{Sueli}.


\section*{Acknowledgment}

The authors would like to thank the Centre Interfacultaire Bernoulli (CIB) at EPFL where part of this work was developed, during the special semester on Combinatorial, Algebraic and Algorithmic Aspects of Coding Theory.



%
\bibliographystyle{plain}
\bibliography{curves}

\end{document}